\documentclass[journal,12pt,onecolumn,draftclsnofoot,]{IEEEtran}
\IEEEoverridecommandlockouts
\usepackage{cite}
\usepackage{amssymb,amsfonts}

\usepackage{graphicx}
\usepackage{subfigure}
\usepackage{textcomp}
\usepackage[tbtags]{amsmath}
\usepackage{diagbox}
\usepackage{booktabs}
\usepackage{multirow}
\usepackage{color}
\usepackage{amsthm} 
\usepackage{algorithm}
\usepackage{algpseudocode}

\alglanguage{pseudocode}
\algblockdefx{Process}{EndProcess}[1]{{\bf process} #1 }{{\bf end process}}
\algblockdefx{Parallel}{EndParallel}[1]{{\bf begin parallel processing} #1 }{{\bf end parallel processing}}
\algblockdefx{FP}{EndFP}[1]%
{\textbf{for all} #1 \textbf{do in parallel}}%
{\textbf{end for}}

\newtheorem{Theorem}{Theorem}
\newtheorem{Lemma}{Lemma}
\newtheorem{Remark}{Remark}

\newtheorem{Corollary}{Corollary}
\newtheorem{Definition}{Definition}
\theoremstyle{definition}

\def\BibTeX{{\rm B\kern-.05em{\sc i\kern-.025em b}\kern-.08em
    T\kern-.1667em\lower.7ex\hbox{E}\kern-.125emX}}

\begin{document}

\title{Coded Computing for Master-Aided Distributed Computing Systems}

\author{
  \IEEEauthorblockN{Haoning Chen and Youlong Wu\\}  
}

\maketitle

\begin{abstract}

 We consider a  MapReduce-type task running in a distributed computing model which consists of ${K}$ edge computing nodes distributed across the edge of the network and a Master node that assists the edge nodes to compute output functions. The Master node and the edge nodes, both equipped with some storage memories and computing capabilities, are connected through a multicast network. We define the communication time spent during the transmission for the sequential implementation (all nodes send symbols sequentially) and parallel implementation (the Master node can send symbols during the edge nodes' transmission), respectively.  We propose a  mixed coded distributed computing scheme that divides the system into two subsystems where the coded distributed computing  (CDC) strategy proposed by Songze Li \emph{et al.} is applied into the first subsystem and a novel master-aided CDC strategy  is applied into the second subsystem.  We prove that  this  scheme is optimal, i.e.,  achieves the minimum communication time for both the sequential and parallel implementation, and establish an {\emph{optimal}} information-theoretic tradeoff between the overall communication time, computation load, and the Master node's storage capacity.   It demonstrates that incorporating a Master node with storage and computing capabilities can further reduce the communication time. For the sequential implementation, we deduce the approximately optimal file allocation between the two subsystems, which shows that  the Master node should  map as many files as possible in order to achieve smaller communication time. For the  parallel implementation, if the Master node's storage and computing capabilities are sufficiently large (not necessary to store and map all files), then the proposed  scheme requires at most 1/2 of the minimum communication time of system without the help of the Master node.
 


\end{abstract}


\section{Introduction}

Mobile edge computing is a distributed network architecture that offloads the computing tasks from the center cloud server to the distributed edge nodes. We focus on a prevalent distributed computing structure called MapReduce \cite{MapReduce}, where the data shuffling time (communication time) often becomes the bottleneck of many MapReduce applications. 

Some previous works have put forward some distributed computing models and used coding strategies to reduce the communication time. In \cite{CDC} Songze Li \emph{et al.}  proposed a scheme called coded distributed computation (CDC) based on the MapReduce framework, and established a fundamental tradeoff between the computation load  and the communication load. In \cite{Optimality}, it studied  resource allocation problem such that the  execution time of the MapReduce-type task is minimized. An asymmetric coded distributed coding (ACDC) was proposed where the   servers are divided into solvers and  helpers.  The helpers  map all the  input files and multicast the message symbols, and the solvers map and decode the message symbols and then generate Reduce functions based on the self-generated and decoded messages. 
Based on a Erd\"os-R\'enyi random graph, a computing network was proposed that achieves an asymptotically inverse-linear tradeoff between the computation load and the average communication load \cite{Graph}. A  similar model allowing every two nodes to exchange data with a probability $p$ was investigated in  \cite{randon_connectivity}. Multiple computations were considered in \cite{Compressed}, where  combined packets  are generated by combining several intermediate results for the single computation and   are  multicast to the network. In \cite{Straggler_0, Straggler_1,Straggler_2,Straggler_3,Straggler_4,Straggler_5,Straggler_6}, the authors focused on the  distributed computing systems with straggling servers.

Note that in many  edge computation systems, such as the federate learning framework \cite{Federated_1,Federated_2,Federated_3,Bonawitz'FL}, the edge nodes are connected with a Master node which is equipped with some storage memories and computing capability and can  help the edges nodes compute  output functions.  This master-aided distributed computing system has been studied in \cite{CCEP,Scalable}. \cite{CCEP} proposed a coded computing architecture where the edge nodes upload the computation tasks to the central node, which is responsible for implementing the coded computations and sending the computation results back to the edge nodes. The proposed scheme can simultaneously minimize the amount of calculation and communication. \cite{Scalable} considered a MapReduce framework for wireless communication scenarios, in which multiple edge nodes are connected wirelessly by an access point. The framework performs redundant calculations on the side of edge nodes  and utilizes coding techniques to reduce the communication time, thereby achieving a scalable design. However, \cite{CCEP} only considered the computing capabilities of the central node, while the edge nodes cannot store any file or do any computation. The central node in \cite{Scalable} is not responsible for any computing task, it only receives information symbols uploaded from the edge nodes and makes linear combinations of them, then sends the linear combinations back to the edge nodes.


In this paper, we consider a master-aided distributed computing systems where   both the Master node and $K$ edge nodes  have a certain storage  and computing capabilities.   The MapReduce-type task has $N$ input files and $Q$ output functions running in this master-aided distributed computing framework. We  define  communication time for the   sequential implementation where all nodes send symbols sequentially and the parallel implementation where the Master node can send symbols during the edge nodes' transmission, respectively.  We then propose a mixed coded computing scheme that divides the system into two subsystems, where in the first subsystem, the CDC scheme is applied to send the required intermediate values generated based on files stored at edge nodes but not at the Master node, and in the second subsystem, a generalized ACDC, extending the original ACDC scheme from case $s=1$ to $s\geq 1$, is applied to send the required intermediate values generated based on the files stored at the Master node.  This scheme is proved to be optimal, i.e.,  achieve the minimum communication time spent during the shuffle phase. An {\emph{optimal}} information-theoretic tradeoff between the overall communication time, computation load, and the Master's nodes storage capacity is established. It demonstrates that incorporating a Master node with computing capabilities can further reduce the communication time. Furthermore, we show that for the parallel implementation, our scheme  improves the communication time of CDC by at least a factor  $\max\{\frac{1}{2}, 1-\frac{M_0}{N}\}$ for the case $s = 1$. 


\section{Problem Formulation}\label{Sec_Model}
We consider a MapReduce-type task running in a  distributed computing system, in which there are a Master node, $K$ edge nodes, $N$ input files $w_1, \ldots, w_N\in\mathbb{F}_{2^F}$, $Q$ output functions $u_q=\phi_q(w_1, \ldots, w_N)\in\mathbb{F}_{2^B}, q\in\{1,\ldots,Q\}$, for some integers $K, N,F,Q,B$ and $1 \leq K \leq N$. The $K$ edge nodes and the Master node are connected through a noiseless broadcast network. The $K$ edge nodes are responsible for computing the $Q$ output functions.  In this paper, we consider the general case where each reduce function is computed by $s$ edge nodes for  $s\in\{1,\ldots,K\}$ so that it can support general computation model with multiple rounds of map and reduce. The Master node, equipped with cache memories that are able to store $M_0\in \mathbb{N}$ files, can assist the edge nodes to generate the output functions. We consider the nontrivial case $M_0\leq N$, as the condition $M_0= N$ is sufficient for the Master node to cache and map all input files. 

The whole computation process is divided into three phases:  Map ,  Shuffle  and  Reduce. 
\subsubsection{Map phase} For $n=\{1,\ldots,N\}$, the Map functions $\overrightarrow{g_n}=(g_{1,n}, \dots, g_{Q,n})$ map the input file $w_n$ into $Q$ length-$T$ intermediate values $v_{q,n}= g_{q,n}(w_n)\in\mathbb{F}_{2^T}$, for some $T\in \mathbb{N}$. We denote the set of files that are mapped at edge node $k$ as $\mathcal{M}_k$, for some  $k\in \{1, \ldots, K\}$ and $\mathcal{M}_k\subseteq \{w_1, \ldots, w_N\}$. For each file $w_n\in\mathcal{M}_k$, edge node $k$ computes  $g_n(w_n)=(v_{1,n}, \ldots, v_{Q,n})$. Since the Master node has a cache size of $M_0$ files, it can map at most  $M_0$ of the $N$ input files. We denote the set of files that are mapped at the Master node as $\mathcal{M}_0$, for some $\mathcal{M}_0 \subseteq \{w_1, \ldots, w_N\}$. For each file $w_n$ in $\mathcal{M}_0$, the Master node computes  $g_n(w_n)=(v_{1,n}, \ldots, v_{Q,n})$.

\begin{Definition} \label{r_def}
(Computation Load): {The computation load, denoted by $r$, for some $0 \leq r \leq K$, is defined as the ratio between the sum of the number of files mapped by the $K$ edge nodes and the total number of files $N$,} i.e., $r\triangleq \frac{\begin{matrix} \sum_{k=1}^K \left|\mathcal{M}_k \right| \end{matrix}}{N}$. 

\end{Definition}
{
\begin{Definition} \label{M1_def}
(Average Map Cost): The average map cost, denoted by $M_1$, is defined as the total number of  mapped files across  the $K$ edge nodes, normalized by $K$, i.e., $M_1\triangleq \frac{\begin{matrix} \sum_{k=1}^K \left|\mathcal{M}_k \right| \end{matrix}}{K}$. The average map cost  $M_1$ can be interpreted as the average number of files that mapped by each edge node.

\end{Definition}

\begin{Remark}
From Definition \ref{r_def} and \ref{M1_def}, we have $M_1\leq N$ and $r = \frac{KM_1}{N}$.  To guarantee that  the whole distributed computation system is able to map all $N$ input files, we must have  $KM_1+M_0 \geq N$.  
\end{Remark}
}

\subsubsection{Shuffle phase}
{
In the Shuffle phase, each edge node $k$, $k\in \{1, \ldots, K\}$, generates a message symbol $X_k = \psi_{k}(\overrightarrow{g_n}:w_n\in\mathcal{M}_k)$, $X_k\in\mathbb{F}_{2^{l_k}}$, for some $l_k \in \mathbb{N}$,   based on its   local intermediate values generated in the Map phase and then
sends it to the other edge nodes to meet their demands. Also, the Master node builds a message symbol $X_{0}\in\mathbb{F}_{2^{l_0}}$, for some $l_{0}\in \mathbb{N}$,   based on its  local intermediate values, 
and then  multicasts it to all the edge nodes.} {By the end of the Shuffle phase, every edge node successfully receives $X_1, \ldots, X_K$ and $X_0$ free of error. 

\begin{Definition}\label{DefLoad}
(Communication Load): 
The communication load is  define  as
\begin{IEEEeqnarray}{rCl}
L&\triangleq& \frac{l_1+ \dots +l_K+l_0}{QNT}.
\end{IEEEeqnarray}
Here $L$ represents the (normalized) total number of bits communicated by the $K$ edge nodes and the Master node during the Shuffle phase.  
\end{Definition} 

In this work, we consider two kinds of implementations: the sequential implementation and the parallel implementation. For the sequential implementation, the Master node and  the $K$ edge nodes send messages sequentially, while for the parallel implementation, $K$ edge nodes and the Master node are able to transmit message symbols simultaneously. 
}

\begin{Definition}\label{DefTime}
(Communication Time): 
Let $L_\textnormal{S}$ and $L_\textnormal{P}$ denote the communication time for the  sequential and parallel implementations, respectively.  We define  
\begin{IEEEeqnarray}{rCl}
L_\textnormal{S}&\triangleq& L= \frac{l_1+ \dots +l_K+l_0}{QNT},\\
L_\textnormal{P}&\triangleq &\max\left\{\frac{l_1+ \dots +l_K}{QNT}, \frac{l_0}{QNT}\right\}.
\end{IEEEeqnarray}
Here $L_\textnormal{P}$ represents the maximum between the (normalized) number of bits sent by the $K$ edge nodes and (normalized) number of bits multicast by the Master node during the Shuffle phase. 

\end{Definition}


\subsubsection{Reduce phase}
We assume that $\frac{Q}{{K\choose s}}\in \mathbb{N}$, and assign reduce tasks symmetrically to $K$ edge nodes.  More precisely, edge node $k$, $k \in \{1, \ldots, K\}$, wishes to compute 
$u_q$ if $q\in \mathcal{W}_k$, where $\mathcal{W}_k\subseteq\{1,\ldots,Q\}$ is a set including the indices of output function assigned to edge node $k$, satisfying: 1) $\left|\mathcal{W}_1 \right|= \dots = \left|\mathcal{W}_K \right|=\frac{Q}{{K\choose s}}$ , 2) $\big(\cap_{j\in \mathcal{S}_1  }\mathcal{W}_j\big) \cap  (\cap_{k\in \mathcal{S}_2  }\mathcal{W}_k\big)=\varnothing$ for all $\mathcal{S}_1,\mathcal{S}_2\subseteq \{1,\ldots,K\}$ if $\mathcal{S}_1\neq \mathcal{S}_2$.

In the Reduce phase, edge node $k\in \{1, \ldots, K\}$ decodes its needed intermediate values based on its local intermediate values and  the messages $X_1, \ldots, X_K$ and $X_0$ received in the Shuffle phase, i.e.,
$(v_{q,1},\ldots,v_{q,N})=d_k^q (X_0,X_1,\ldots,X_k, \{\overrightarrow{g_n}:w_n\in\mathcal{M}_k)\}),$
where $d_k^q$ is a decoding function for edge node $k$ and Reduce function $q\in\mathcal{W}_k$.  
  Edge node $k$ then uses the intermediate values $(v_{q,1},\ldots,v_{q,N})$ to compute the Reduce functions $u_q$.
%

A  computation-communication tuple $(r,s,L_\textnormal{S})$ for the sequential implementation is feasible if for any positive $\epsilon$, and sufficiently large $N$, there exist $\{\mathcal{M}_k\}_{k=0}^K$, $\{\mathcal{W}_k\}_{k=1}^K$, a set of  encoding functions $\{\psi_{k}\}_{k=1}^K$, and a set of decoding function $\{d^q_k,k\in\mathcal{W}_k\}_{k=1}^K$ to  achieve the computation-communication tuple $(\tilde{r},\tilde{s},\tilde{L}_\textnormal{S})$ such that $|r-\tilde{r}|\leq \epsilon$,  $|s-\tilde{s}|\leq \epsilon$, $|L_\textnormal{S}-\tilde{L}_\textnormal{S}|\leq \epsilon$, and Node $k\in\{1,\ldots,K\}$ can successfully compute all the output functions $(u_q:  q\in \mathcal{W}_k)$.

\begin{Definition}
The computation-communication function for the distributed computing framework with $s\geq1$ is define as
\begin{IEEEeqnarray}{rCl}
L^*_\textnormal{S}(r,s,M_0) \triangleq \textnormal{inf} ~\{L_\textnormal{S}: (r,s,M_0,L_\textnormal{S})~\textnormal{is feasible} \}.
\end{IEEEeqnarray}
\end{Definition}
We simply use $L^*_\textnormal{S}(r,M_0)$ to denote $L^*_\textnormal{S}(r,s,M_0)$ when $s=1$.  For the  parallel implementation, we define the computation-communication function $L^*_\textnormal{P}(r,s,M_0) $ similar as before, but with the subscript $\textnormal{S}$ replaced by $\textnormal{P}$.

\section{Main Results}\label{Sec_Results}
Let 
\begin{subequations} \label{EqL1} 
\begin{IEEEeqnarray}{rCl}
L_1(r_1,s)  &\triangleq & \mathop{\sum}\limits_{\ell = \max\{r_1+1,s\}}^{\min\{r_1+s,K\}} \frac{\binom{K}{\ell} \binom{\ell-1}{r_1} \binom{r_1}{\ell-s}}{ \binom{K}{r_1} \binom{K}{s}}  \frac{\ell}{\ell - 1}, \\
L_2(r_2, s)  &\triangleq & \mathop{\sum}\limits_{\ell = \max\{r_2+1, s\}}^{\min\{r_2+s, K\}} \frac{\binom{K}{\ell} \binom{\ell - 1}{r_2} \binom{r_2}{\ell - s}}{\binom{K}{r_2} \binom{K}{s}},
\end{IEEEeqnarray}
\end{subequations}
for any $r_1,r_2 \in \{0, \ldots, K\}$. By letting $s=1$ in \eqref{EqL1}, we have $L_1(r_1,s)=L_1(r_1)$ and $L_1(r_1,s)=L_1(r_1)$ where
\begin{subequations}
\begin{IEEEeqnarray}{rCl}
L_1(r_1)  &\triangleq &  \frac{1}{r_1}  \left(1-\frac{r_1}{K}\right),\\
L_2(r_2)  &\triangleq &  \frac{1}{r_2+1}  \left(1-\frac{r_2}{K}\right).
\end{IEEEeqnarray}
\end{subequations}


\begin{Theorem} \label{LS(rs)}

The computation-communication function of the overall distributed computing system for the sequential implementation with $s\geq 1$, $L_\textnormal{S}^*(r, s,M_0)$, is given by  
\begin{IEEEeqnarray}{rCl}\label{EqSeq}
L_\textnormal{S}^*(r, s,M_0)
=\!\!\!\!\!\!\!\!\! \min_{\substack{(\alpha, r_1,r_2):\\
 \alpha  r_1 + (1 - \alpha )  r_2=r,\\r_1,r_2\in\{0,\ldots,K\},\\
 \alpha \in [1-\frac{M_0}{N},1]
  } }  \!\!\!\!\!\!\!\!\!
  \alpha   L_1^*(r_1, s) + (1 - \alpha)  L_2^*(r_2, s),~
\end{IEEEeqnarray}
for some $r\in[0,K]$, 
 where  $L_1^*(r_1, s)$ is the lower convex envelope of the points $\{(r_1, L_1(r_1,s)): r_1\in \{0, \dots, K\}\}$, and  $L_2^*(r_2, s)$ is the lower convex envelope of the points $\{(r_2, L_2(r_2,s)): r_2\in \{0, \dots, K\}\}$.
\end{Theorem}

\begin{proof}
We divide the overall system into two subsystems. Subsystem 1 is responsible for mapping $ \alpha  N$ input files, Subsystem 2 is responsible for mapping $(1- \alpha ) N$ input files, for some $ (1-\frac{M_0}{N}) \leq \alpha  \leq 1$. The computation load of Subsystem 1 is $r_1$, and the computation load of Subsystem 2 is $r_2$, for some $0 \leq r_1 \leq K$ and $0 \leq r_2 \leq K$. The communication time of Subsystem 1 is $L_1(r_1,s)$, and the computation time of Subsystem 2 is $L_2(r_2,s)$, for some $0 \leq L_1 \leq 1$, $0 \leq L_2 \leq 1$.
See more details in Section \ref{Scheme}. The converse proof is given in Appendix \ref{Converse}. 
\end{proof}

\begin{Remark}
When $s=1$, the coding scheme for Subsystem 2 is essentially identical to the coded caching scheme introduced in \cite{Centralized}, where a central server connects through a shared and error-free link to $K$ edge nodes. In the placement phase, each of the $K$ edge nodes with cache size $M_1$ prefetches some files, in the delivery phase, the central server broadcasts coded messages needed by the edge nodes through the shared link, thereby meeting all the edge nodes' requests. When the number of files $N$ and the number of edge nodes $K$ satisfy $N \geq K$, this scheme achieves a minimum rate (communication load) of $R_C(M_1) = K \cdot (1 - \frac{M_1}{N}) \cdot \frac{1}{1 + \frac{KM_1}{N}}$. We let $r = \frac{KM_1}{N}$, and the normalized minimum rate (communication time) is $L = \frac{R_C(M_1)}{K} = \frac{1}{r+1} \cdot (1 - \frac{r}{K})$, which matches the minimum communication time of Subsystem 2. 
\end{Remark}

Let $s=1$ in  \eqref{EqSeq}, then we obtain the following corollary.
\begin{Corollary} \label{LS(r)}
	The computation-communication function of the overall distributed computing system for the sequential implementation with $s=1$, $L_\textnormal{S}^*(r,M_0)$, is given by 
\begin{align}\label{EqLoads1}
L_\textnormal{S}^*(r,M_0)
=\!\!\!\!\!\! \min_{\substack{(\alpha, r_1,r_2):\\
 \alpha  r_1 + (1 - \alpha )  r_2=r,\\r_1,r_2\in\{0,\ldots,K\},\\
 \alpha \in [1-\frac{M_0}{N},1]
  } }  \!\!\!\!\!\!
  \alpha L_1^*(r_1) + (1 - \alpha) L_2^*(r_2),
\end{align}
 for  some $r\in[0,K]$,     where  $L_1^*(r_1)$ is the lower convex envelope of the points $\{(r_1, L_1(r_1)) : r_1\in \{0, \dots, K\}\}$, and  $L_2^*(r_2)$ is the lower convex envelope of the points $\{(r_2, L_2(r_2)) : r_2\in \{0, \dots, K\}\}$.
 \end{Corollary}
 
 \begin{Remark}
 When ignoring the integrality constraints in \eqref{EqLoads1}, we prove in Appendix \ref{Optimal_S}  that   the optimal $\alpha$, $r_1$, and  $r_2$ can be approximated as follows: $\alpha^* = (1 - \frac{M_0}{N})$, $r_1^* = \frac{r + 1 -(1 - \frac{M_0}{N})}{\sqrt{1 + \frac{1}{K}} - \left(\sqrt{1 + \frac{1}{K}} - 1 \right) \cdot (1 - \frac{M_0}{N})},\ \\r_2^* = \frac{\sqrt{1 + \frac{1}{K}}r -(1 - \frac{M_0}{N})}{\sqrt{1 + \frac{1}{K}} - \left(\sqrt{1 + \frac{1}{K}} - 1 \right) \cdot (1 - \frac{M_0}{N})} $. {The choice of $\alpha^*$ reflects that the overall communication time for sequential implementation will be minimized if the Master node stores  and maps as many files as possible. In particular when $M_0=N$, i.e., the Master node is sufficiently large to store (map) all the files in the system, the optimal solutions are $\alpha^* = 0$, $r_1^*=0$ and  $r^*_2=r$.}
 \end{Remark}

\begin{Theorem} \label{LP(rs)}

The computation-communication function of the overall distributed computing system for the parallel implementation when $s\geq 1$, $L_\textnormal{P}^*(r, s,M_0)$, is given by
\begin{IEEEeqnarray}{rCl}\label{EqPara}
L_\textnormal{P}^*(r,s,M_0)= \!\!\!\!\!\! \min_{\substack{(\alpha, r_1,r_2):\\
 \alpha  r_1 + (1 - \alpha )  r_2=r,\\r_1,r_2\in\{0,\ldots,K\},\\
 \alpha \in [1-\frac{M_0}{N},1]
  } }  \!\!\!\!\!\!
  \max \{\alpha   L_1^*(r_1, s), (1 \! - \!  \alpha)  L_2^*(r_2, s)\},\nonumber\\
\end{IEEEeqnarray}
 for some $r\in[0,K]$,    
 where  $L_1^*(r_1, s)$ is the lower convex envelope of the points $\{(r_1, L_1(r_1,s)): r_1\in \{0, \dots, K\}\}$, and  $L_2^*(r_2, s)$ is the lower convex envelope of the points $\{(r_2, L_2(r_2,s)): r_2\in \{0, \dots, K\}\}$.

\end{Theorem}
\begin{proof}
We divide the  system into two subsystems in the same way as the scheme for  Theorem \ref{LS(rs)}, and then let the two subsystems send messages symbols simultaneously, leading to a computation-communication function in \eqref{EqPara}. See more details in Section \ref{Scheme}. The converse proof is given in Appendix \ref{Converse}.
\end{proof} 
Let $s=1$ in  \eqref{EqPara}, then we obtain the following corollary.
\begin{Corollary} \label{LP(r)}
	The computation-communication function of the overall distributed computing system for the parallel implementation with $s=1$, $L_\textnormal{P}^*(r)$, is given by 
\begin{IEEEeqnarray}{rCl}
L_\textnormal{P}^*(r,M_0)
= \!\!\!\!\!\! \min_{\substack{(\alpha, r_1,r_2):\\
 \alpha  r_1 + (1 - \alpha )  r_2=r,\\r_1,r_2\in\{0,\ldots,K\},\\
 \alpha \in [1-\frac{M_0}{N},1]
  } }  \!\!\!\!\!\!
  \max\{  \alpha L_1^*(r_1), (1- \alpha) L_2^*(r_2)\}, \quad~
\end{IEEEeqnarray}
for some $r\in[0,K]$,  where  $L_1^*(r_1)$ is the lower convex envelope of the points $\{(r_1, L_1(r_1) ): r_1\in \{0, \dots, K\}\}$, and  $L_2^*(r_2)$ is the lower convex envelope of the points $\{(r_2, L_2(r_2) ): r_2\in \{0, \dots, K\}\}$.
\end{Corollary}

{
 \begin{Corollary}
 For the parallel implementation with $s=1$, the minimum communication time is upper bounded by 
\begin{IEEEeqnarray}{rCl} \label{EqUpper}
L^*_\textnormal{P}(r,M_0) &\leq & \max\{\frac{1}{2}, 1-\frac{M_0}{N}\}  L_1(r).
\end{IEEEeqnarray}
If the Master node's cache size $M_0$ is sufficiently large such that $M_0\geq \frac{r+1}{2r+1}N $, we have
\begin{IEEEeqnarray}{rCl}
L^*_\textnormal{P}(r,M_0) \leq \frac{L_1(r)L_2(r)}{L_1(r)+L_2(r)}=\frac{1}{2r+1}\left(1-\frac{r}{K}\right).
\end{IEEEeqnarray}

 \end{Corollary}
 \begin{proof}
 See proof in Appendix \ref{LPrange}.
 \end{proof}
 
\begin{Remark}
The upper bound of $L_\textnormal{P}^*(r,M_0)$ in \eqref{EqUpper} shows that   the  parallel scheme achieves at most $\max\{\frac{1}{2}, 1-\frac{M_0}{N}\}$ of   communication time of CDC  for the case $s = 1$. 
\end{Remark}

 Fig. 1 depicts the relationship between the computation load and the minimum communication time for the uncoded scheme, CDC scheme, and our proposed schemes for both the sequential implementation and the parallel implementation. We can see that our schemes achieve smaller communication time than the CDC scheme under the same computation load. 


\section{Illustrative Examples}\label{Examples}
\begin{figure}
	\centering
	\includegraphics[width=0.65\textwidth,scale=0.2]{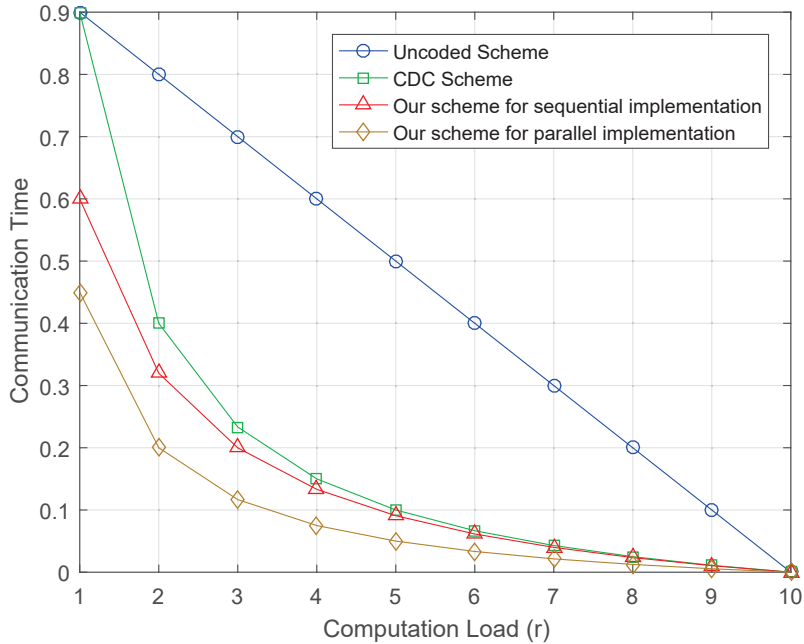} 
	\caption{Illustration of the relationship between the computation load $r$ and the communication time $L$ for the uncoded scheme, the  CDC scheme, and our proposed scheme with a Master node under the sequential implementation and the parallel implementation in a distributed computing system consisting of $K=10$ edge nodes, $N=2520$ files and $Q=10$ output functions.} 
	\label{fig_result}
\end{figure}
Consider a MapReduce-type task running in the distributed computing system consisting $N$ input files, $K=3$ edge nodes and a Master node. The 3 edge nodes are responsible for computing $Q=3$ output functions. Each output function here is only mapped by one edge node, i.e., $s=1$. 
\;\\
\textbf{Example 1 ($M_0 = N=6,$ $M_1=4$)}: In this case, the Master node can map all input files. As  shown in Fig. \ref{fig:a}, the three output functions are represented by the red circle, green square and blue triangle respectively. {Edge node 1, 2, and 3 wish to compute Reduce functions represented by the red circle, the green square, and the blue triangle, respectively.} Since $M_1 = 4$, the system has a computation load of $r=2$ (i.e., each file is mapped by 2 different edge nodes). Each edge node makes full use of its cache memories to map 4 different input files. Meanwhile, the Master node maps all the 6 files in the system. After the Map phase, each file is transformed into 3 corresponding intermediate values, each of which is needed by an output function. 

Since there are $N=6$ files in total, the output function in each edge node needs 6 intermediate values to compute their reduce functions. After the Map phase, each edge node gets 4 intermediate values, and the remaining 2 intermediate values have to be obtained from the outside. For example, edge node 1 knows the red circle in File 1, 2, 3, 4, and it requires red circle in File 5 and File 6 from the outside. Since the Master node maps all the files in the system, it knows all the intermediate values required by the edge nodes after the Map phase. Therefore, during the Shuffle phase, the Master node generates 2 bit-wise XORs, denoted by $\oplus$, each of which is a combination of 3 intermediate values needed by the edge nodes, and then broadcasts them to all the 3 edge nodes. Each edge node has already obtained 2 of the 3 intermediate values in each bit-wise XOR after the Map phase, thus it can decode the rest one, which is necessary for its reduce computations. For example, given that edge node 1 knows the green square in File 1 and the blue triangle in File 3, it can decode the desired the red circle in File 5 from the message symbol broadcast by the Master node. Thus, the communication time is $\frac{2}{3 \times 6} = \frac{1}{9}$. Compared to the CDC scheme, which  yields a communication time of $\frac{1}{6}$, the proposed scheme improves the communication time by a factor of $\frac{2}{3}$ with the help of the Master node.  

\begin{figure}
\centering
\subfigure[]
{\label{fig:a} 
\includegraphics[width=0.47\linewidth,scale=1]{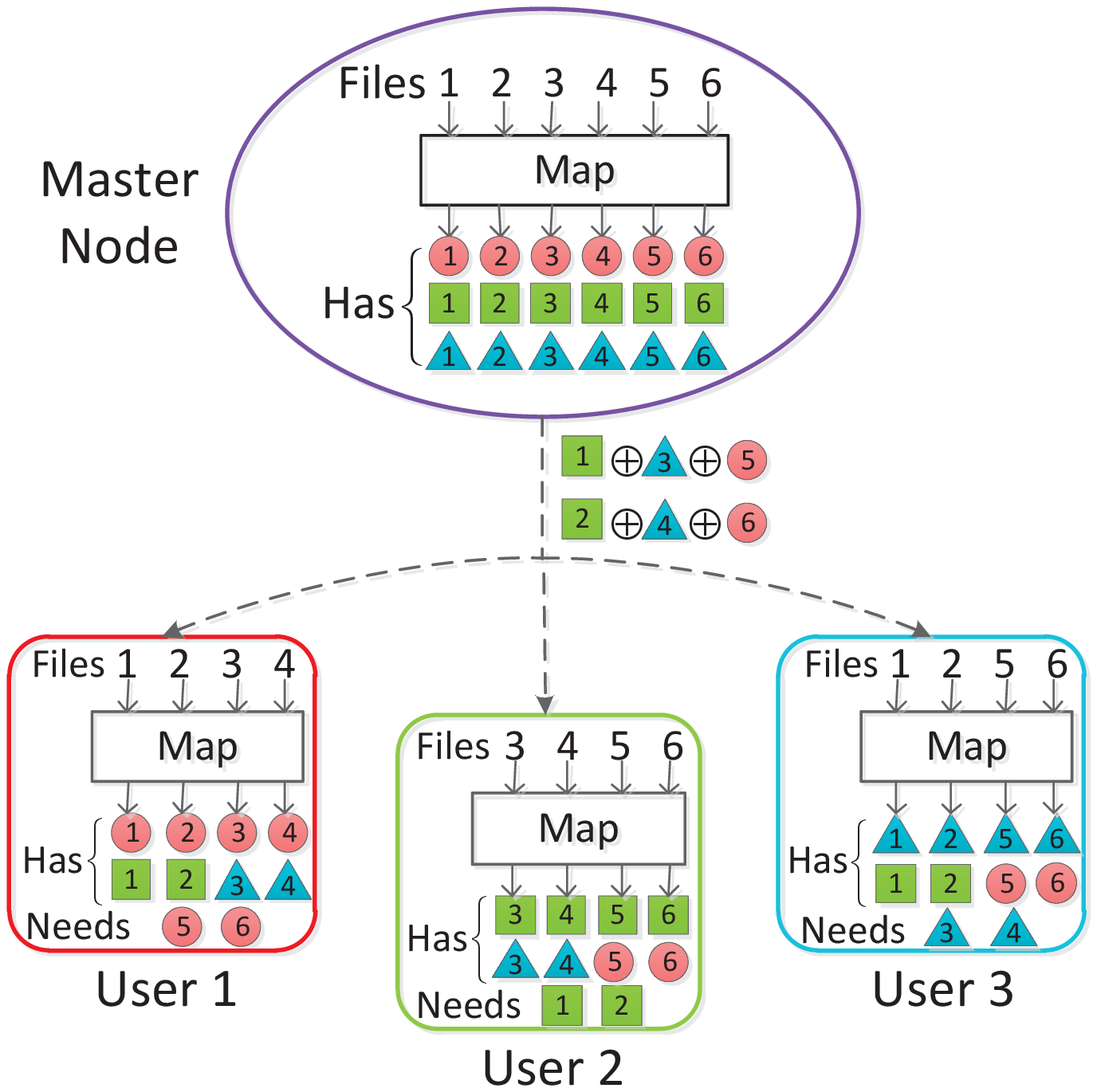}} 
\vfill
\subfigure[]
{\label{fig:b} 
\includegraphics[width=0.6\linewidth,scale=0.01]{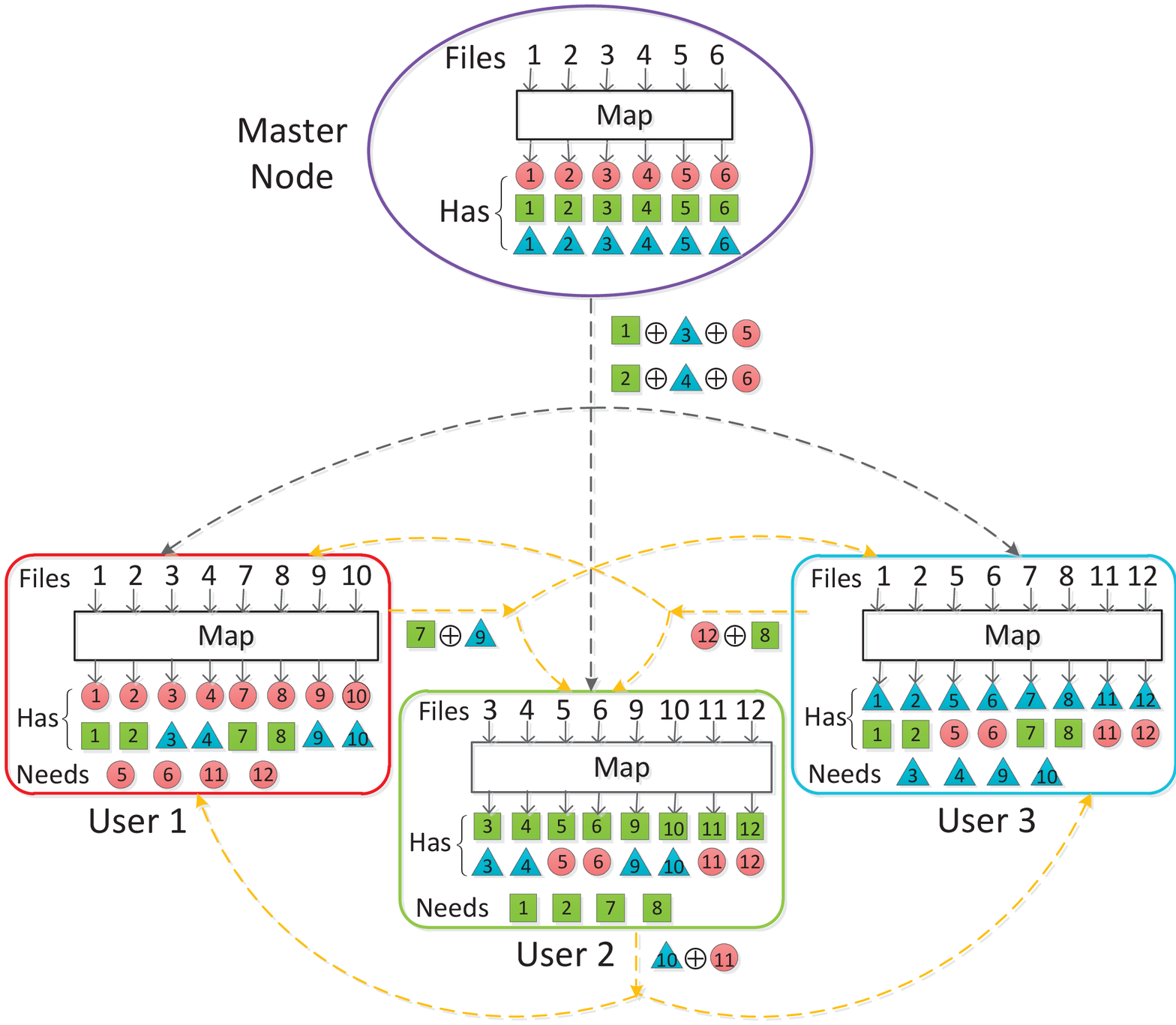}} 
\caption{Illustrations of a distributed computing system responsible for computing $Q=3$ output functions from $N$ input files, using $K = 3$ edge nodes and a Master node. (a) $M_0=N$ (when the Master node is large enough to map all the files). (b) $M_0=\frac{N}{2}$ (when the Master node can only map half of the files).
}
\label{figb} 
\end{figure}

\noindent \textbf{Example 2} ($N=12, M_0 = 6, M_1=8$): In this case,  the Master node cannot map all  input files. As is shown in Fig.  \ref{fig:b}, each edge node maps  $M_1=8$ files, thus the system still has a computation load of $r=2$. In Fig. \ref{fig:b}}, the Master node only maps file 1 to 6, while file 7 to 12 are mapped exclusively by the 3 edge nodes. Therefore,  the locally mapped intermediate values and the messages symbols sent by the Master nodes are not sufficient for the  edge nodes to produce the assigned reduce functions. These edge nodes have to exchange some of the needed intermediate values with each other. The transmitted message symbols are given in  Fig.  \ref{fig:b}. After receiving the message symbols sent by the Master node and other edge nodes, it can be checked that every edge node can decode the desired intermediate values. For the sequential implementation where the Master node and the edge nodes send message symbols sequentially, the resulting communication time is $\frac{2+3}{3 \times 12} = \frac{5}{36}$; for the parallel implementation where $K$ edge nodes and the Master node transmit intermediate values simultaneously, the resulting communication time is $\frac{\max\{2, 3\}}{3 \times 12} = \frac{1}{12}$.

\section{General Achievable Scheme} \label{Scheme}
Divide the overall distributed computing system into two subsystems. 
Subsystem 1 is responsible for mapping $N_1 \triangleq \alpha N$ input files, for some $(1 - \frac{M_0}{N}) \leq \alpha \leq 1$. Each edge node is responsible for mapping $M_1'$ of the $N_1$ input files in Subsystem 1, for some $M_1'\in\mathbb{N}$ with $0 \leq M_1' \leq M_1$. We denote the subset of files in Subsystem 1 that are mapped by edge node $k$ as $\mathcal{M}_k^1$, for some $k \in \{1, \ldots, K\}$, $\mathcal{M}_k^1 \subseteq \{w_1, \ldots, w_{N_1}\}$ and $|\mathcal{M}_k^1| = M_1'$. 
Subsystem 2 is responsible for mapping the remaining $N_2 \triangleq (1-\alpha) N$ input files. Each edge node is responsible for mapping $M_1 - M_1'$ of the $N_2$ input files in Subsystem 2. We denote the subset of files in Subsystem 2 that are mapped by edge node $k$ as $\mathcal{M}_k^2$, for some $k \in \{1, \ldots, K\}$, $\mathcal{M}_k^2 \subseteq \{w_{N_1+1}, \ldots, w_{N_1+N_2}\}$ and $|\mathcal{M}_k^2| = M_1 - M_1'$. Moreover, the Master node maps all the $N_2$ input files in Subsystem 2. 
Specially, we let $M_1'=0$ if $\alpha=0$ and $M_1'=M_1$ if $\alpha = 1$.

We define the computation load and communication time of the two subsystems similar to Definition \ref{r_def} and \ref{DefTime}.

\begin{Definition}\label{Defr1r2}
The computation load of Subsystem 1, denoted by $r_1$, for some $0 \leq r_1 \leq K$, is defined as the  total number of files in Subsystem 1 mapped by the $K$ edge nodes normalized by the total number of files in Subsystem 1, i.e., $r_1\triangleq \frac{ \sum_{k=1}^K \left|\mathcal{M}_k^1 \right| }{N_1} = \frac{KM_1'}{N_1} = \frac{KM_1'}{\alpha N}$.  If $\alpha = 0$, let   $r_1=0$.

The computation load of Subsystem 2, denoted by $r_2$, for some $0 \leq r_2 \leq K$, is defined as the total number of files in Subsystem 2 mapped by the $K$ edge nodes normalized by the total number of files in Subsystem 2, i.e., $r_2\triangleq \frac{ \sum_{k=1}^K \left|\mathcal{M}_k^2 \right| }{N_2} = \frac{K(M_1-M_1')}{N_2} = \frac{K(M_1-M_1')}{(1-\alpha) N}$. If $\alpha = 1$, let   $r_2=0$.

\end{Definition}

\begin{Definition}
 The communication load of  Subsystem 1 and Subsystem 2, denoted by $L_1$ and $L_2$, respectively,  are defined as 
 \begin{IEEEeqnarray}{rCl}
&& L_1 \triangleq \frac{ \sum_{k=1}^K l_k }{QN_1T} = \frac{ \sum_{k=1}^K l_k }{\alpha QNT},\\
 L&&_2 \triangleq \frac{l_0}{QN_2T} = \frac{l_0}{(1-\alpha) QNT}.
\end{IEEEeqnarray}
\end{Definition}
From Definition \ref{r_def} and \ref{Defr1r2}, we know that the overall computation load $r$ is a convex combination of $r_1$ and $r_2$, the computation load of the two subsystems, i.e., 
\begin{IEEEeqnarray}{rCl}
r = \alpha r_1+(1 - \alpha) r_2.
\end{IEEEeqnarray}

We consider the case where  $N$ is sufficiently large such that  for any $\epsilon>0$, we always have $|\alpha N-\lfloor \alpha N\rfloor|\leq \epsilon$, and $|(1-\alpha) N-{K\choose t}\eta_1|\leq \epsilon$  
for any $\alpha \in [1-\frac{M_0}{N},1]$, $t\in\{1,\ldots,K\}$ and some $\eta_1\in\mathbb{N}$.  
Next, we present achievable schemes for the two subsystems respectively. 

\subsection{Subsystem 1}

The MapReduce process of Subsystem 1 is identical to that in the Coded Distributed Computing (CDC) scheme introduced in \cite{CDC}. Therefore, the communication time of this subsystem is 
\begin{IEEEeqnarray}{rCl} \label{L_1} 
L_1(r_1,s) = \mathop{\sum}\limits_{\ell = \max\{r_1+1,s\}}^{\min\{r_1+s,K\}} \frac{\binom{K}{\ell} \binom{\ell-1}{r_1} \binom{r_1}{\ell-s}}{ \binom{K}{r_1} \binom{K}{s}} \cdot \frac{\ell}{\ell - 1}, 
\end{IEEEeqnarray}
for any $r_1 \in \{1, \ldots, K\}$. 

\begin{Remark} \label{L1, s=1}
By letting  $s=1$ in  \eqref{L_1}, i.e.,  each output function is only computed by a single edge node, the communication load of Subsystem 1 is $\frac{1}{r_1} \cdot (1-\frac{r_1}{K})$, for any $r_1 \in \{1, \ldots, K\}$. 
\end{Remark}

\subsection{Subsystem 2}\label{SecSubsystem2}
We first discuss the case when $r_2$ is a positive integer, i.e., $r_2 \in \{1, \ldots, K\}$, and then apply the result to a special case when $r_2=0$. 

When $r_2 = K$, 
each file in Subsystem 2 is mapped by all the $K$ edge nodes and thus there is no need for the Master node to broadcast any message to the edge nodes, resulting in $L_2^*(K, s) = 0$ for all $s \in \{1, \ldots, K\}$. 

\subsubsection{Map phase}
The process in the Map phase is same as that in the CDC scheme, differing in that 1) here we have a dataset of ${N}_2 = (1 - \alpha)N$ input files rather than the whole  $N$ input files; 2) there exists a Master node mapping all ${N}_2$ input files; 3) each file is mapped by $r_2$ edge nodes and the Master node.

Specifically, we divide the $N_2$ files evenly into $\binom{K}{r_2} $ disjoint groups, each group contains $\eta_{1}$ files and corresponds to a subset $\mathcal{T}$ of size $r_2$: 
\begin{IEEEeqnarray}{rCl}
\{w_{N_1+1}, \ldots, w_{N_1+N_2}\} = \mathop{\bigcup}\limits_{\mathcal{T} \subset \{1, \ldots, K\}, \left| \mathcal{T} \right| = r_2} \mathcal{B}_{\mathcal{T}}, 
\end{IEEEeqnarray} 
where $\mathcal{B}_{\mathcal{T}}$ denotes the group of files mapped exclusively by the edge nodes in $\mathcal{T}$ and the Master node. Each edge node $k \in \{1, \ldots, K\}$  maps all files in $\mathcal{B}_{\mathcal{T}}$ if $k \in \mathcal{T}$, and generates local intermediate values $\{v_{q,n} : q \in \{1, \ldots, Q\}, w_n \in \mathcal{M}_k^2\}$ with $\left| \mathcal{M}_k^2 \right| = \binom{K-1}{r_2-1} \eta_{1} = \frac{r_2 N_2}{K}$.


\subsubsection{Shuffle phase}
The process in the Shuffle phase is similar to that in the CDC scheme, mainly differing in  that 1) in our scheme only the Master node  generates and multicasts message symbols, while in CDC scheme it's the  edge nodes to send message symbols; 2) in CDC scheme each intermediate value is split into segments, which is not necessary in our scheme; 3) to achieve a larger multicast gain, our scheme generates each message symbol  as a linear combination of $n_1\triangleq\binom{\left| \mathcal{S} \right|}{r_2}$ elements, while the value of $n_1$ in CDC is smaller, equaling to $\binom{\left| \mathcal{S}\right|-1}{r_2-1}$. 

Specifically, we divide $Q$ reduce functions evenly into $\binom{K}{s} $ disjoint groups, i.e., each group contains $\eta_{2}$ files and corresponds to a subset $\mathcal{P}$ of size $s$: 
\begin{IEEEeqnarray}{rCl} 
\{1, \ldots, Q\} = \mathop{\bigcup}\limits_{\mathcal{P} \subset \{1, \ldots, K\}, \left| \mathcal{P} \right| = s} \mathcal{D}_{\mathcal{P}}, 
\end{IEEEeqnarray} 
where $\mathcal{D}_{\mathcal{P}}$ denotes the group of reduce functions computed exclusively by the edge nodes in $\mathcal{P}$. Given the allocation above, for each edge node $k \in \{1, \ldots, K\}$, if $k \in \mathcal{P}$, it maps all the files in $\mathcal{D}_{\mathcal{P}}$. Each edge node is in $\binom{K-1}{s-1}$ subsets of size $s$, thus it is responsible for computing $\left| \mathcal{W}_k \right| = \binom{K-1}{s-1} \eta_{2} = \frac{sQ}{K}$ Reduce functions, for all $k \in \{1, \ldots, K\}$.

For a subset $\mathcal{S} \subseteq \{1, \ldots, K\}$, and a subset $\mathcal{S}_1 \subset \mathcal{S} : \left| \mathcal{S}_1 \right| = r_2$, {denote the set of intermediate values required by all edge nodes in $\mathcal{S} \setminus \mathcal{S}_1$  while exclusively known by both the Master node and  $r_2$ edge nodes in $\mathcal{S}_1$ as $V_{\mathcal{S}_1}^{\mathcal{S} \setminus \mathcal{S}_1}$, } i.e.,
\begin{align} \label{eqV}
V_{\mathcal{S}_1}^{\mathcal{S} \setminus \mathcal{S}_1} &\triangleq \{v_{q,n}: q \in \mathop{\cap}\limits_{k \in {\mathcal{S} \setminus \mathcal{S}_1}} \mathcal{W}_k, q \notin \mathop{\cup}\limits_{k \notin \mathcal{S} } \mathcal{W}_k, 
\notag
\\&
\phantom{=\;\;}
\quad \quad \quad  w_n \in  \mathop{\cap}\limits_{k \in  {\mathcal{S}_1}} \mathcal{M}_k^2, w_n \notin \mathop{\cup}\limits_{k \notin \mathcal{S}_1} \mathcal{M}_k^2 \}.
\end{align}

Following the same analysis in \cite{CDC}, we know that $V_{\mathcal{S}_1}^{\mathcal{S} \setminus \mathcal{S}_1}$     contains $\binom{r_2}{\left| \mathcal{S} \right| - s}\eta_2\eta_1$ intermediate values for   $\binom{r_2}{s \ - \ (\left| \mathcal{S} \right| - r_2)}\eta_2 = \binom{r_2}{\left| \mathcal{S} \right| - s}\eta_2$ output functions. 

Create a   symbol $ U_{\mathcal{S}_1}^{\mathcal{S}\setminus\mathcal{S}_1} \in \mathbb{F}_{2^{\binom{r_2}{\left| \mathcal{S} \right| - s}\eta_2\eta_1 T}}$ by concatenating all  intermediate values in $ V_{\mathcal{S}_1}^{\mathcal{S}\setminus\mathcal{S}_1}$. Given the set $\mathcal{S}$, there are in total $n_1 \triangleq \binom{\left| \mathcal{S} \right|}{r_2}$ subsets each with cardinality equaling to $r_2$. Index these $n_1$ set as $\mathcal{S}[1],\ldots,\mathcal{S}[n_1]$, and the corresponding message symbols are  $U_{\mathcal{S}[1]}^{\mathcal{S} \setminus \mathcal{S}[1]}, U_{\mathcal{S}[2]}^{\mathcal{S} \setminus \mathcal{S}[2]}, \ldots,  
\\ U_{\mathcal{S}[n_1]}^{\mathcal{S} \setminus \mathcal{S}[n_1]}$. After the Map phase, since the Master node knows all the intermediate values needed by the edge nodes, it multicasts $n_2 \triangleq \binom{\left| \mathcal{S} \right| - 1}{r_2}$ linear combinations of the $n_1$ message symbols to the edge nodes in $\mathcal{S}$, denoted by $X_0^{\mathcal{S}} \left[ 1 \right], X_0^{\mathcal{S}} \left[ 2 \right], \ldots, X_0^{\mathcal{S}} \left[ n_2 \right], $ for some coefficients $\alpha_1, \alpha_2, \ldots, \alpha_{n_1}$ distinct from one another and $\alpha_i \in \mathbb{F}_{2^{\binom{r_2}{\left| \mathcal{S} \right| - s}\eta_1 \eta_2 T}} $ for all $i = 1, \ldots, n_2$, i.e., 
\begin{align}
\begin{bmatrix}
X_0^{\mathcal{S}}[1]\\
X_0^{\mathcal{S}}[2]\\
\vdots\\
X_0^{\mathcal{S}}[n_2]
\end{bmatrix}
=
\begin{bmatrix}
1 & 1 & \cdots & 1 \\
\alpha_1 & \alpha_2 & \cdots & n_1 \\
\vdots & \vdots & \ddots & \vdots \\
\alpha_1^{n_2-1} & \alpha_2^{n_2-1} & \cdots & \alpha_{n_1}^{n_2-1} \\
\end{bmatrix}
\begin{bmatrix}
U_{\mathcal{S}[1]}^{\mathcal{S} \setminus \mathcal{S}[1]}\\
U_{\mathcal{S}[2]}^{\mathcal{S} \setminus \mathcal{S}[2]}\\
\vdots\\
U_{\mathcal{S}[n_1]}^{\mathcal{S} \setminus \mathcal{S}[n_1]}
\end{bmatrix}.
\end{align}
Since each edge node $k \in \mathcal{S}$ is in $\binom{\left| \mathcal{S} \right| - 1}{r_2 - 1}$ subsets of $\mathcal{S}$ with size $r_2$, it knows $\binom{\left| \mathcal{S} \right| - 1}{r_2 - 1}$ of the  symbols $U_{\mathcal{S}[1]}^{\mathcal{S} \setminus \mathcal{S}[1]}, U_{\mathcal{S}[2]}^{\mathcal{S} \setminus \mathcal{S}[2]}, \ldots, U_{\mathcal{S}[n_1]}^{\mathcal{S} \setminus \mathcal{S}[n_1]}$, and given that $\alpha_1, \alpha_2, \ldots, \alpha_{n_1}$ are distinct from each other, it can decode the rest of $n_1 - \binom{\left| \mathcal{S} \right| - 1}{r_2 - 1} = \binom{\left| \mathcal{S} \right|}{r_2} - \binom{\left| \mathcal{S} \right| - 1}{r_2 - 1} = \binom{\left| \mathcal{S} \right| - 1}{r_2} = n_2$ message symbols. 

Therefore, in the Shuffle phase, for each subset $\mathcal{S} \subseteq \{1, \ldots, K\}$ of size $\max\{r_2+1, s\} \leq \left| \mathcal{S} \right| \leq \min\{r_2+s, K\}$, the Master node multicasts $ n_2 = \binom{\left| \mathcal{S} \right| - 1}{r_2}$ message symbols to the edge nodes in $\mathcal{S}$, each message symbol containing $\binom{r_2}{\left| \mathcal{S} \right| - s}\eta_1 \eta_2 T $ bits. Therefore, the communication time of Subsystem 2 achieved by the above scheme is 
\begin{align}\label{L_2}
L_2(r_2, s) &= \mathop{\sum}\limits_{\ell = \max\{r_2+1, s\}}^{\min\{r_2+s, K\}} \frac{\binom{K}{\ell} \binom{\ell - 1}{r_2} \binom{r_2}{\ell - s} \eta_1 \eta_2 T}{QN_2T} \nonumber
\\&= \mathop{\sum}\limits_{\ell = \max\{r_2+1, s\}}^{\min\{r_2+s, K\}} \frac{\binom{K}{\ell} \binom{\ell - 1}{r_2} \binom{r_2}{\ell - s}}{\binom{K}{r_2} \binom{K}{s}}.
\end{align}
\begin{Remark} \label{L2, s=1}
By letting $s=1$ in \eqref{L_2},   i.e.,  each output function is only computed by a single edge node,  the communication load of Subsystem 2 is $\frac{1}{r_2+1}  (1-\frac{r_2}{K})$, for any $r_2 \in \{0, \ldots, K\}$. 
\end{Remark}

When $r_2=0$, the $K$ edge nodes do not map any file in Subsystem 2. Therefore, in the Shuffle phase, the Master node simply unicasts needed intermediate values to  edge nodes. Since the $K$ edge nodes require $Q N_2 T$ intermediate values in total, this case yields a communication time of $L_2 (0,s)= 1$.

\subsection{Overall Communication Time}


Therefore, for the sequential implementation, according to the definition of $L_\textnormal{S}$, the overall communication time of the entire system is 
\begin{align} \label{EqLS}
L_\textnormal{S}&= \frac{ Q N_1 T \cdot L_1 + Q N_2 T \cdot L_2}{QNT} \nonumber
\\&= \frac{ Q \alpha N T \cdot L_1 + Q (1-\alpha) N T \cdot L_2}{QNT} \nonumber
\\&= \alpha \cdot \mathop{\sum}\limits_{\ell = \max\{r_1+1,s\}}^{\min\{r_1+s,K\}} \frac{\binom{K}{\ell} \binom{\ell-1}{r_1} \binom{r_1}{\ell-s}}{\binom{K}{r_1} \binom{K}{s}} \cdot \frac{\ell}{\ell - 1} \nonumber \\ &\quad +  (1-\alpha) \cdot \mathop{\sum}\limits_{\ell = \max\{r_2+1, s\}}^{\min\{r_2+s, K\}} \frac{\binom{K}{\ell} \binom{\ell - 1}{r_2} \binom{r_2}{\ell - s}}{\binom{K}{r_2} \binom{K}{s}}, 
\end{align} 
for some $r_1,r_2\in\{0,\ldots,K\}$ and $\alpha r_1+(1-\alpha) r_2=r$.
For the parallel implementation, according to the definition of $L_\textnormal{P}$, the overall communication time of the entire system is 
\begin{align} 
L_\textnormal{P}\nonumber &= \max\{\frac{Q N_1 T \cdot L_1}{QNT}, \frac{Q N_2 T \cdot L_2}{QNT}\} \nonumber
\\&= \max\{\frac{Q \alpha N T \cdot L_1}{QNT}, \frac{Q (1-\alpha) N T \cdot L_2}{QNT}\} \nonumber
\\&= \max \{ \alpha \cdot \mathop{\sum}\limits_{\ell = \max\{r_1+1,s\}}^{\min\{r_1+s,K\}} \frac{\binom{K}{\ell} \binom{\ell-1}{r_1} \binom{r_1}{\ell-s}}{\binom{K}{r_1} \binom{K}{s}} \cdot \frac{\ell}{\ell - 1}, \nonumber 
\\& \; \quad \quad \; \; \; \; \; \; \; \; (1-\alpha) \cdot \mathop{\sum}\limits_{\ell = \max\{r_2+1, s\}}^{\min\{r_2+s, K\}} \frac{\binom{K}{\ell} \binom{\ell - 1}{r_2} \binom{r_2}{\ell - s}}{\binom{K}{r_2} \binom{K}{s}} \}, 
\end{align} 
for some $r_1,r_2\in\{0,\ldots,K\}$ and $\alpha r_1+(1-\alpha) r_2=r$.

\begin{Remark} 
When $s=1$ (i.e., each output function is computed by one edge node), the overall communication time for the sequential implementation is $ \alpha \cdot \frac{1}{r_1} \cdot (1-\frac{r_1}{K}) + (1-\alpha) \cdot \frac{1}{r_2+1} \cdot (1-\frac{r_2}{K})$, and the overall communication time for parallel implementation is $ \max\{\alpha \cdot \frac{1}{r_1} \cdot (1-\frac{r_1}{K}), (1-\alpha) \cdot \frac{1}{r_2+1} \cdot (1-\frac{r_2}{K})\}$. 
\end{Remark}

\section{Conclusions and Future Direction}\label{Sec_Conclusion}

In this paper, we considered the  Master-aided distributed computing systems where the Master node is capable of mapping files and sending symbols to help edge noes compute the output functions. We proposed a mixed coded distributed computing scheme that achieves the minimum communication time, for both the sequential implementation and the parallel implementation. Compared with the coded distributed computing scheme without the help of the  Master node,  the proposed scheme  can further reduce the communication time. For the sequential implementation, we deduce the approximate optimal file allocation between the two subsystems, which shows that  the Master node should  map as many files as possible in order to achieve smaller communication time.  When the Master node's storage and computing capabilities are sufficiently large, not necessary to store and map all files,  the proposed parallel scheme achieves at most 1/2 of the minimum communication time of system without the help of the Master node.  
 


\appendices

\section{Choice of $(r_1, r_2, \alpha)$ for the Sequential Implementation with $s=1$}\label{Optimal_S}

We consider the the sequential implementation with   $s=1$ and derive the approximately optimal $\alpha$ and $M_1'$ omitting the constraints that $\alpha N, M_1', r_1, r_2 \in \mathbb{N}$, denoted by $\alpha^*$ and $M_1'^*$. 


Substituting $s = 1$ into \eqref{EqLS}, we have  
\begin{align}\label{EqLs1}
L_\textnormal{S} &
= \alpha \cdot \frac{1}{r_1} \cdot (1-\frac{r_1}{K})+ (1-\alpha) \cdot \frac{1}{r_2+1} \cdot (1-\frac{r_2}{K}) \nonumber
\\&= \alpha \cdot \frac{1}{r_1} + (1-\alpha) \cdot \frac{1+\frac{1}{K}}{r_2+1} - \frac{1}{K}. 
\end{align}
Since $r_1 = \frac{KM_1'}{\alpha N}$ and $r_2 = \frac{K(M_1-M_1')}{(1-\alpha) N}$, \eqref{EqLs1} can be written as 
\begin{align}\label{EqLs2} 
L_\textnormal{S} &= \alpha \cdot \frac{1}{\frac{KM_1’}{\alpha N}} + (1-\alpha) \cdot \frac{1+\frac{1}{K}}{\frac{K(M_1-M_1')}{(1-\alpha) N}+1} - \frac{1}{K} \nonumber
\\&= \frac{N \alpha^2}{KM_1'} + \frac{(1+\frac{1}{K}) N (1-\alpha)^2}{K(M_1 - M_1') + (1-\alpha)N} - \frac{1}{K}.
\end{align}
According to \eqref{EqLs2}, the first derivative of $L_\textnormal{S}$ respect to $M_1’$ is 
\begin{align}\label{EqLs3} 
\frac{\mathrm{d}L_\textnormal{S}}{\mathrm{d}M_1'} &= -\frac{N \alpha^2}{KM_1’^2} + \frac{(K+1) N (1-\alpha)^2}{[K(M_1 - M_1') + (1-\alpha)N]^2}. 
\end{align} 
Let $\frac{\mathrm{d}L_\textnormal{S}}{\mathrm{M_1'}} = 0$, then we have
\begin{align}\label{beta*} 
M_1'^* = \frac{[KM_1 + (1-\alpha)N] \cdot \alpha}{\sqrt{K} [\sqrt{K} \cdot \alpha + \sqrt{K+1} \cdot (1-\alpha)]}.
\end{align} 
By \eqref{EqLs3}, the second derivative of $L_\textnormal{S}$ respect to $M_1'$ is 
\begin{align} 
\frac{\mathrm{d}^2 L_\textnormal{S}}{\mathrm{d}M_1'^2} &= \frac{2N \alpha^2}{KM_1'^3} + \frac{2K(K+1) N (1-\alpha)^2}{[K(M_1 - M_1') + (1-\alpha)N]^3}.
\end{align} 
It is easy to see that $\frac{\mathrm{d}^2 L_\textnormal{S}}{\mathrm{d}M_1'^2} > 0$. Therefore, $M_1'^*$ is the only optimal solution of $M_1'$ to minimize  $L_\textnormal{S}$. 
 
Substituting \eqref {beta*} into \eqref {EqLs2}, the simplified relation between $L_\textnormal{S}$ and $\alpha$ is 
\begin{align}\label{EqLs4} 
L_\textnormal{S} =\frac{N}{K} \cdot \frac{[(\sqrt{K+1} - \sqrt{K}) \cdot \alpha - \sqrt{K+1}]^2}{KM_1+(1-\alpha)N} - \frac{1}{K}. 
\end{align} 
According to \eqref{EqLs4}, the first derivative of $L_\textnormal{S}$ respect to $\alpha$ is 
\begin{align}
\frac{\mathrm{d} L_\textnormal{S}}{\mathrm{d} \alpha} &= \frac{N}{K} \cdot \left\{ \frac{2(\sqrt{K+1}-\sqrt{K})[KM_1+(1-\alpha)N]}{[KM_1+(1-\alpha)N]^2}
\notag\right. 
\\&
\phantom{=\;\;}
\left. \; \cdot \ [(\sqrt{K+1}-\sqrt{K}) \alpha - \sqrt{K+1}] 
\notag\right. 
\\&
\phantom{=\;\;}
\left.
+ \frac{[(\sqrt{K+1}-\sqrt{K}) \alpha - \sqrt{K+1}]^2N}{[KM_1+(1-\alpha)N]^2} \right\}.
\end{align}

We let 
\begin{align}\label{EqLs5}
f&(\alpha) = 2(\sqrt{K+1}-\sqrt{K}) \cdot [(\sqrt{K+1}\!-\!\sqrt{K}) \alpha\! - \!\sqrt{K+1}] \nonumber
\\ &\cdot[KM_1+(1\!-\!\alpha)N] + [(\sqrt{K+1}\!-\!\sqrt{K}) \alpha \!-\! \sqrt{K+1}]^2N.
\end{align} 

Note that the quadratic coefficient in \eqref{EqLs5} is $-(\sqrt{K+1}-\sqrt{K})^2N$, which is a negative number. Therefore,  the second derivative of $L_\textnormal{S}$ respect to $\alpha$ is negative, i.e., $f''(\alpha) < 0$. 

When $\alpha = 0$, \eqref{EqLs5} becomes
\begin{align}
f(0) = [2(\sqrt{K+1}-\sqrt{K})(KM_1+N) &- \sqrt{K+1}\ N] \nonumber
\\& \cdot ( - \sqrt{K+1}).
\end{align}

Since 
\begin{align}
&\frac{2(\sqrt{K+1}-\sqrt{K})(KM_1+N)}{ \sqrt{K+1}\ N} \nonumber 
\\ \overset{(a)}=\ &2 \left(1 - \sqrt{1 - \frac{1}{K+1}}\right) \cdot (r+1) \nonumber
\\=\ &\frac{2(r+1)}{(K+1)\left(1+\sqrt{1-\frac{1}{K+1}}\right)} \nonumber 
\\ \overset{(b)} < \ &\frac{2(r+1)}{(K+1)\left(1+1-\frac{1}{K+1}\right)} \nonumber 
\\=\ &\frac{2(r+1)}{2K+1} \nonumber
\\ \approx \ &1, \nonumber
\end{align}
where (a) is due to $r=\frac{KM_1}{N}$, and (b) is due to $\sqrt{1-\frac{1}{K+1}} > 1-\frac{1}{K+1}$. As a result, we have $2(\sqrt{K+1}-\sqrt{K})(KM_1+N) - \sqrt{K+1}\ N <0$, and thus $f(0) \geq 0$. 

When $\alpha = 1$, \eqref{EqLs5} becomes
\begin{align}
f(1) = [2(\sqrt{K+1}-\sqrt{K})KM_1 - &\sqrt{K}\ N]  \cdot ( - \sqrt{K}). 
\end{align}

Since 
\begin{align}
&\frac{2(\sqrt{K+1}-\sqrt{K})\ KM_1}{\sqrt{K}\ N} \nonumber 
\\ \overset{(a)}=\ &2 \left(\sqrt{1+\frac{1}{K}} - 1\right) \cdot r \nonumber
\\ \overset{(b)} < \ &\frac{r}{K} \nonumber 
\\ \leq \ &1, \nonumber
\end{align}
where (a) is due to $r=\frac{KM_1}{N}$, and (b) is due to $\sqrt{1+\frac{1}{K}} \ - 1 < \frac{1}{2K}$. As a result, we have $2(\sqrt{K+1}-\sqrt{K})KM_1 - \sqrt{K}\ N <0$, and thus $f(1) \geq 0$. 

In summary, since $f(0) \geq 0$, $f(1) \geq 0$, and when $0 \leq \alpha \leq 1$, $f''(\alpha) < 0$, thus for any $0 \leq \alpha \leq 1$, we have $f(\alpha) \geq 0$. Therefore, when $0 \leq \alpha \leq 1$, $\frac{\mathrm{d} L_\textnormal{S}}{\mathrm{d} \alpha} \geq 0$, i.e., the smaller $\alpha$ is, the smaller $L_\textnormal{S}$ will be. Since the range of $\alpha$ is $(1 - \frac{M_0}{N}) \leq \alpha \leq 1$, the optimal solution of $\alpha$ is 
\begin{equation} \label{alpha*}
\alpha^* = 1 - \frac{M_0}{N}. 
\end{equation}

Therefore, the optimal solutions of $\alpha$ and $M_1'$ for this case are shown in \eqref{alpha*} and \eqref{beta*}, which result in the minimum communication time.

\section{Upper bounds of $L_\textnormal{P}^*$} \label{LPrange}
 For simplicity, we use notation $L_\textnormal{P}^*$ for  $L_\textnormal{P}^*(r,M_0)$.
 Rewrite the definitions $L_1(r_1) \triangleq \frac{1}{r_1} \cdot (1-\frac{r_1}{K})$ and $L_2(r_2) \triangleq \frac{1}{r_2+1} \cdot (1-\frac{r_2}{K})$. Obviously, we have $L_1(r) > L_2(r)$. We choose the computation load of the two subsystems to be equal, i.e., $r_1 = r_2 = r$, then
\begin{align} \label{LPupperbound}
L_\textnormal{P}^*&\leq \min_{\alpha\in[1-M_0/N,1]}\ \max\{\alpha L_1(r), (1-\alpha) L_2(r)\} \nonumber
\\&\leq \min_{\alpha\in[1-M_0/N,1]}\ \max\{\alpha L_1(r), (1-\alpha) L_1(r)\}. 
\end{align}

If $M_0 \geq \frac{N}{2}$, since $\alpha \geq (1 - \frac{M_0}{N})$, $\alpha$ can be $\frac{1}{2}$, then \eqref{LPupperbound} becomes 
\begin{align*} 
L_\textnormal{P}^*&\leq \min_{\alpha\in[1-M_0/N,1]}\ \max\{\alpha , (1-\alpha) \} L_1(r)= \frac{1}{2} L_1(r). 
\end{align*} 

If $M_0 < \frac{N}{2}$, then  $\alpha\geq (1 - \frac{M_0}{N}) >\frac{1}{2}$ and $1-\alpha\leq \frac{M_0}{N} <\frac{1}{2}$, thus  \eqref{LPupperbound} becomes 
\begin{align*} 
L_\textnormal{P}^*&\leq \min_{\alpha\in[1-M_0/N,1]}\! \max\{\alpha , (1\!-\!\alpha) \}L_1(r)= (1\!-\!\frac{M_0}{N}) L_1(r). 
\end{align*} 

Therefore, the upper bound of $L_\textnormal{P}^*$ that we derive is 
\begin{align}
L_\textnormal{P}^*\leq \max\{\frac{1}{2}, 1-\frac{M_0}{N}\} L_1(r). 
\end{align}

Again, by letting $r_1=r_2=r$,  we have 
\begin{align} \label{LPlowerbound}
L_\textnormal{P}^*&\leq \min_{\alpha}\ \max\{\alpha L_1(r), (1-\alpha) L_2(r)\}.
\end{align}

If $M_0 $ is sufficiently large such that there exists $\alpha^*$ satisfying 
$\alpha^*\geq 1-\frac{M_0}{N}$ and 
$\alpha^* L_1(r)=(1-\alpha^*) L_2(r)$, which is equivalent to $$ \quad \alpha^*=\frac{L_2(r)}{L_1(r)+L_2(r)}=\frac{r}{2r+1},\ M_0 \geq \frac{r+1}{2r+1}{N},$$  then we have
\begin{align} 
L_\textnormal{P}^*&\leq  \min_{\alpha}\ \max\{\alpha L_1(r), (1-\alpha) L_2(r)\} = \frac{L_1(r)L_2(r)}{L_1(r)+L_2(r)}. \nonumber
\end{align}

\section{Converse Proofs of Theorem \ref{LS(rs)} and \ref{LP(rs)}} \label{Converse} 

In this section, we prove the lower bound of $L_\textnormal{S}^*(r, s)$ and $L_\textnormal{P}^*(r, s)$ proposed in Theorem \ref{LS(rs)} and Theorem \ref{LP(rs)}. 

We first present a lemma 
proved in \cite{Optimality}, with only a minor change replacing $K$ with $K+1$.

\begin{Lemma}\label{Lemma1}
Consider a MapReduce-type task  with $N$ files and $Q$ reduce functions, and a given map and reduce design that runs in  a distributed computing system consisting of $K+1$ computing nodes. For any integers $c$, $d$,  let $a_{c,d}$ denote the number of intermediate values that are available at $c$ nodes, and required by (but not available at) $d$ nodes. The following lower bound on the communication load holds: \begin{equation} L \geq \frac{1}{QN}{\sum_{c=1}^{K+1} \sum_{d=1}^{K-c+1} a_{c,d} \frac{d}{c+d-1}}. \end{equation}

\end{Lemma}

Recall that $\mathcal{M}_k$, for $k\in[1,\ldots,K]$, and $\mathcal{M}_0$  represent the sets of mapped files by edge node  $k$, and the Master node, respectively.  
Given a map design $\{\mathcal{M}_k\}_{k=0}^K$ and reduce design  $\{\mathcal{W}_k\}_{k=1}^K$, we define  
\begin{IEEEeqnarray}{rCl}
\mathcal{N}_1\triangleq \{w_1,\ldots,w_N\}\backslash \mathcal{M}_0,~\text{and}~\alpha \triangleq \frac{|\mathcal{N}_1|}{N},
\end{IEEEeqnarray} 
then we have $1-\alpha= \frac{|\mathcal{M}_0|}{N}$.

 Let $b_{j,1}$ denote the number of files that are mapped by $j$ edge nodes, but not mapped by the Master node, and let $b_{j,2}$ be the number of files that are mapped by $j$ edge nodes, and also mapped by the Master node. By the definitions of $\alpha$, $b_{j,1}$ and $b_{j,2}$, and since $r=\sum_{k=1}^K |\mathcal{M}_k|/N$, we have 
\begin{subequations}
\begin{IEEEeqnarray}{rCl}
&b_{j,1} \geq 0 \label{Converse_1_0}, 
\\& b_{j,2} \geq 0 \label{Converse_1_1}, 
\\&\alpha = \mathop{\sum}\limits_{j=0}^K \frac{b_{j,1}}{N} \label{Converse_1_3}, 
\\&1-\alpha = \mathop{\sum}\limits_{j=0}^K \frac{b_{j,2}}{N} \label{Converse_1_4}, 
\end{IEEEeqnarray}
\newpage 
\begin{IEEEeqnarray}{rCl} \nonumber
\\& \mathop{\sum}\limits_{j=0}^K (b_{j,1} + b_{j,2}) = N \label{Converse_1_2},
\\& \mathop{\sum}\limits_{j=0}^K j(b_{j,1} + b_{j,2}) = rN \label{Converse_1_2_1}.  
\end{IEEEeqnarray}
\end{subequations}

We proceed to define
\begin{align} 
&r_1 \triangleq \mathop{\sum}\limits_{j=0}^K \frac{jb_{j,1}}{\alpha N} \label{Converse_1_5}, 
\\&r_2 \triangleq \mathop{\sum}\limits_{j=0}^K \frac{jb_{j,2}}{(1-\alpha) N} \label{Converse_1_6},
\end{align} 
and from the definitions of $r_1$, $r_2 $   and $r$, we have
\begin{IEEEeqnarray}{rCl}
\alpha r_1+(1-\alpha)r_2 =r.
\end{IEEEeqnarray}

According to Lemma \ref{Lemma1}, we have
\begin{IEEEeqnarray}{rCl}
L^* &\geq &\frac{1}{QN}{\sum_{c=1}^{K+1} \sum_{d=1}^{K-c+1} a_{c,d} \frac{d}{c+d-1}},\nonumber\\
&=& \frac{1}{QN}{\sum_{c=1}^K \sum_{d=1}^{K-c} n_{c,1} \frac{d}{c+d-1}}+\frac{1}{QN}{\sum_{c=1}^{K+1} \sum_{d=1}^{K-c+1} n_{c,2} \frac{d}{c+d-1}}, \nonumber\\
&=& \frac{1}{QN} \mathop{\sum}\limits_{j=0}^{K} \mathop{\sum}\limits_{d=\max\{1,s-j\}}^{\min\{s,K-j\}} \left(n_{j,1} \frac{d}{j\!+\!d\!-\!1} + n_{j,2} \frac{d}{j+d}\right)\nonumber\\
&=&\alpha L_1^* +(1-\alpha) L_2^*,
\end{IEEEeqnarray}
where 
\begin{IEEEeqnarray}{rCl}
n_{j,1}&\triangleq&\frac{Q}{\binom{K}{s}}b_{j,1} \binom{K-j}{d} \binom{j}{j+d-s}\\
n_{j,2}&\triangleq&\frac{Q}{\binom{K}{s}}b_{j,2} \binom{K-j}{d} \binom{j}{j+d-s}\\
L_1^*&\triangleq& \frac{1}{\alpha NQ} \mathop{\sum}\limits_{j=0}^{K} \mathop{\sum}\limits_{d=\max\{1,s-j\}}^{\min\{s,K-j\}} n_{j,1} \frac{d}{j\!+\!d\!-\!1}, \label{eqL1}\\
L_2^*&\triangleq& \frac{1}{(1-\alpha) NQ} \mathop{\sum}\limits_{j=0}^{K} \mathop{\sum}\limits_{d=\max\{1,s-j\}}^{\min\{s,K-j\}} n_{j,2} \frac{d}{j\!+\!d}\label{eqL2}.
\end{IEEEeqnarray}
Here $n_{j,1}$ is the number of intermediate values that are known by $j$ edge nodes, not known by the super edge node, and needed by (but not available at) $d$ edge nodes; $n_{j,2}$ is the number of intermediate values that are known by $j$ edge nodes and the Master node, and needed by (but not available at) $d$ edge nodes. 

From the definitions of $n_{j,1}$ and $b_{j,1}$, we  know that $L^*_1$ is only related to files in $\mathcal{N}_1$ and $K$ edge nodes, thus $L^*_1$
 can be viewed as the lower bound of a communication load of a distributed system consisting of $\alpha N$ input files and $K$ computing nodes
 are only related to files in $\mathcal{N}_1$. Similarly, $L^*_2$  is only related to files in $\mathcal{M}_0$ and is incurred by  the Master node and $K$ edge nodes, thus $L^*_2$ can be viewed as the lower bound of a communication load of a distributed system consisting $(1-\alpha)N$ input files and $K+1$ computing nodes. Since the Master node is not involved in $L_1^*$, if the communication load   $L_2^*$ is  incurred only by    the Master node's transmission, 
by Definition \ref{DefTime}, we obtain that 
\begin{IEEEeqnarray*}{rCl}
L^*_\textnormal{S}&=&{\alpha L^*_1  +(1-\alpha)L^*_2 },\\
L^*_\textnormal{P}&= &  \max \{\alpha L^*_1, (1-\alpha)L^*_{2}\}.
\end{IEEEeqnarray*}

In the following two subsections, we derive the lower bound on $L^*_1$ and $L^*_2$, respectively, and showed that  $L^*_2$  can  be achieved by and only by the Master node's transmission.

\subsection{The lower bound of $L_1$} 



{
Letting $\ell=j+d$ and from   \eqref{eqL1},  we have 

\begin{align} \label{EqeqLoadUser}
L_1^* 
&= \frac{1}{\binom{K}{s}} \mathop{\sum}\limits_{j=0}^{K} \mathop{\sum}\limits_{\ell=\max\{j+1,s\}}^{\min\{j+s,K\}} \frac{b_{j,1}}{\alpha N} \binom{K-j}{\ell-j} \binom{j}{\ell-s} \frac{\ell-j}{\ell-1} \nonumber
\\& \stackrel{(a)}{\geq} \frac{1}{\binom{K}{s}} \mathop{\sum}\limits_{\ell=\max\{\mathop{\sum}\limits_{j=0}^{K}{\frac{jb_{j,1}}{\alpha N}}+1,s\}}^{\min\{\mathop{\sum}\limits_{j=0}^{K}{\frac{jb_{j,1}}{\alpha N}}+s,K\}} \binom{K-\mathop{\sum}\limits_{j=0}^{K}{\frac{jb_{j,1}}{\alpha N}}}{\ell-\mathop{\sum}\limits_{j=0}^{K}{\frac{jb_{j,1}}{\alpha N}}}\nonumber
\\& \hspace{12ex} \cdot  \binom{\mathop{\sum}\limits_{j=0}^{K}{\frac{jb_{j,1}}{\alpha N}}}{\ell-s} \frac{\ell-\mathop{\sum}\limits_{j=0}^{K}{\frac{jb_{j,1}}{\alpha N}}}{\ell-1} \nonumber
\\& \stackrel{(b)}= \frac{1}{\binom{K}{s}} \mathop{\sum}\limits_{\ell=\max\{r_1+1,s\}}^{\min\{r_1+s,K\}} \binom{K-r_1}{\ell-r_1} \binom{r_1}{\ell-s} \frac{\ell-r_1}{\ell-1} \nonumber
\\& = \mathop{\sum}\limits_{\ell=\max\{r_1+1,s\}}^{\min\{r_1+s,K\}} \frac{\binom{K}{\ell} \binom{\ell-1}{r_1} \binom{r_1}{\ell-s}}{\binom{K}{r_1} \binom{K}{s}} \cdot \frac{\ell}{\ell\! - \!1},
\end{align} 
where ($a$) holds by Jensen's inequality {and because the function $\binom{K-j}{\ell-j} \binom{j}{\ell-s} \frac{\ell-j}{\ell-1} $ is convex in $j$, and $\mathop{\sum}\limits_{j=0}^K \frac{b_{j,1}}{\alpha N}=1$ by \eqref{Converse_1_3}; equality ($b$) holds by \eqref{Converse_1_5}.}
}
For general $0 \leq r_1 \leq K$, we find the line $p_1+q_1j$ as a function of $0 \leq r_1 \leq K$ connecting the two points $(\lfloor r_1 \rfloor, L_1(\lfloor r_1 \rfloor, s))$ and $(\lceil r_1 \rceil, L_1(\lceil r_1 \rceil, s))$, i.e., 
\begin{IEEEeqnarray}{rCl}
&p_1+q_1j|_{j=\lfloor r_1 \rfloor} = L_1(\lfloor r_1 \rfloor, s), 
\\& p_1+q_1j|_{j=\lceil r_1 \rceil} = L_1(\lceil r_1 \rceil, s).
\end{IEEEeqnarray}

Then by the convexity of the function $L_1(j, s)$  in $j$, we have for integer-valued $j =0, \ldots, K$, 
\[L_1(j, s) \geq p_1+q_1j. \]
Therefore, the minimum communication load  $L_1^*$ is lower bounded by   the lower convex envelope of the points $\{(r_1, L_1({r_1}, s)): r_1 \in \{0, \ldots, K\} \}$. 

\subsection{The lower bound of $L_2$} 

Letting $\ell=j+d$ and from \eqref{eqL2}, we have
\begin{IEEEeqnarray}{rCl}\label{eqLoadMaster}
L_2^*&=&  \mathop{\sum}\limits_{j=0}^{K} \mathop{\sum}\limits_{\ell=\max\{j+1,s\}}^{\min\{j+s,K\}} \frac{b_{j,2}}{\binom{K}{s}(1-\alpha) N} \binom{K-j}{\ell-j} \binom{j}{\ell\!-\!s} \frac{\ell\!-\!j}{\ell} \nonumber 
\\\nonumber
\\& \stackrel{(a)}{\geq}& \frac{1}{\binom{K}{s}} \mathop{\sum}\limits_{\ell=\max\{\mathop{\sum}\limits_{j=0}^{K}{\frac{jb_{j,2}}{(1-\alpha) N}}+1,s\}}^{\min\{\mathop{\sum}\limits_{j=0}^{K}{\frac{jb_{j,2}}{(1-\alpha) N}}+s,K\}} \binom{K-\mathop{\sum}\limits_{j=0}^{K}{\frac{jb_{j,2}}{(1-\alpha) N}}}{\ell-\mathop{\sum}\limits_{j=0}^{K}{\frac{jb_{j,2}}{(1-\alpha) N}}} \nonumber
\\&& \hspace{15ex} \ \cdot \binom{\mathop{\sum}\limits_{j=0}^{K}{\frac{jb_{j,2}}{(1-\alpha) N}}}{\ell-s} \frac{\ell-\mathop{\sum}\limits_{j=0}^{K}{\frac{jb_{j,2}}{(1-\alpha) N}}}{\ell} \nonumber
\\& \stackrel{(b)}=& \frac{1}{\binom{K}{s}} \mathop{\sum}\limits_{\ell=\max\{r_2+1,s\}}^{\min\{r_2+s,K\}} \binom{K-r_2}{\ell-r_2} \binom{r_2}{\ell-s} \frac{\ell-r_2}{\ell} \nonumber
\\& =& \mathop{\sum}\limits_{\ell=\max\{r_2+1,s\}}^{\min\{r_2+s,K\}} \frac{\binom{K}{\ell} \binom{\ell-1}{r_2} \binom{r_2}{\ell-s}}{\binom{K}{r_2} \binom{K}{s}},
\end{IEEEeqnarray}
where ($a$) holds by Jensen's inequality {and because  the function $\binom{K-j}{\ell-j} \binom{j}{\ell-s} \frac{\ell-j}{\ell} $ is convex in $j$, and  $\mathop{\sum}\limits_{j=0}^K \frac{b_{j,2}}{(1-\alpha) N}=1$ by \eqref{Converse_1_4};  equality ($b$) holds by \eqref{Converse_1_6}.}

For general $0 \leq r_2 \leq K$, we find the line $p_2+q_2j$ as a function of $0 \leq r_2 \leq K$ connecting the two points $(\lfloor r_2 \rfloor, L_2(\lfloor r_2 \rfloor, s))$ and $(\lceil r_2 \rceil, L_2(\lceil r_2 \rceil, s))$, i.e., 
\begin{align}
&p_2+q_2j|_{j=\lfloor r_2 \rfloor} = L_2(\lfloor r_2 \rfloor, s), 
\\& p_2+q_2j|_{j=\lceil r_2 \rceil} = L_2(\lceil r_2 \rceil, s).
\end{align} 
Then by the convexity of the function $L_2(j, s)$ in $j$, we have for integer-valued $j =0, \ldots, K$, 
\begin{align}
& L_2(j, s) \geq p_2+q_2j. 
\end{align}

Therefore, the minimum communication load $L_2^*$ is lower bounded by  the lower convex envelope of the points $\{(r_2, L_2({r_2}, s)): r_2 \in \{0, \ldots, K\} \}$.

\subsection{Lower Bounds of the Overall Communication Time}

Next, for both the sequential implementation and the parallel implementation, we respectively deduce the lower bound of the overall communication time from the above results. 

Since the overall computation load $r$ satisfies $r = \alpha r_1 + (1 - \alpha) r_2$, for a certain $(1 - \frac{M_0}{N}) \leq \alpha' \leq 1$ and any $0 \leq r' \leq K$, there exist $0 \leq r_1' \leq K$ and  $0 \leq r_2' \leq K$ that satisfy $r' = \alpha r_1' + (1 - \alpha) r_2'$. In other words, any computation load $r$ in $0 \leq r \leq K$ is achievable given appropriate $r_1$ and $r_2$.

In Section \ref{SecSubsystem2}, we showed that the lower bound $L_2^*$ in \eqref{eqLoadMaster} can be achieved by a scheme where the Master node  multicasts some message symbols to the edge nodes. Comparing with the lower bound of communication load in \eqref{EqeqLoadUser} (only edge nodes can send message symbols) and in \eqref{eqLoadMaster} (the Master node can send message symbols) for $r_1=r_2=r$, and since $\frac{\ell}{\ell - 1}>1$ for all $l>1$, we know that sharing any amount of transmission of message symbol from the Master node to the edge nodes will cause strictly larger communication load, thus $L_2^*$ must be incurred solely by the Master node, otherwise the resulting communication load must be strictly larger than $L_2^*$. Therefore, for the sequential and parallel implementations, by Definition \ref{DefTime},  
 the minimum communication time $L_\textnormal{S}^*(r, s)$ and $L_\textnormal{P}^*(r,s)$ is  
\begin{align}
L_\textnormal{S}^*(r, s) &= \alpha \cdot L_1^*(r_1, s) + (1-\alpha) \cdot L_2^*(r_2, s), \\
L_\textnormal{P}^*(r, s) &= \max\{\alpha \cdot  L_1^*(r_1, s), (1-\alpha) \cdot L_2^*(r_2, s)\}, 
\end{align}
satisfying
\begin{equation}
\begin{aligned}
&r = \alpha r_1 + (1 - \alpha) r_2, 
\\& (1 - \frac{M_0}{N}) \leq \alpha \leq 1,
\\& 0 \leq r_1 \leq K\ (r_1 = 0\ \text{if\ and\ only\ if}\ \alpha = 0),
\\& 0 \leq r_2 \leq K.
\end{aligned}
\end{equation}

\end{document}